\documentclass{amsart}
\usepackage{amsmath}
\usepackage{amssymb}
\usepackage{bbm}
\usepackage{color}

\let\set\mathbbm
\def\deg{\operatorname{deg}}
\def\lc{\operatorname{lc}}
\def\lclm{\operatorname{lclm}}
\def\<#1>{\langle#1\rangle}

\newtheorem{example}{Example}
\newtheorem{definition}[example]{Definition}
\newtheorem{theorem}[example]{Theorem}
\newtheorem{lemma}[example]{Lemma}

\begin{document}

\title{Desingularization of Ore Operators}

\author{Shaoshi Chen}
\address{Shaoshi Chen, KLMM, AMSS, Chinese Academy of Sciences, 100190 Beijing, China}
\email{schen@amss.ac.cn}

\author{Manuel Kauers}
\address{Manuel Kauers, Research Institute for Symbolic Computation, J. Kepler University Linz, Austria}
\email{mkauers@risc.jku.at}

\author{Michael F. Singer}
\address{Michael F. Singer, Department of Mathematics, North Carolina State University, Raleigh, NC, USA}
\email{singer@math.ncsu.edu}

\thanks{S.C. was supported by the NSFC grant 11371143 and a 973 project (2011CB302401), M.K. was supported by FWF grant Y464-N18, and M.F.S. was supported by NSF grant CCF-1017217.}



\begin{abstract}
  We show that Ore operators can be desingularized by calculating a least
  common left multiple with a random operator of appropriate order. Our result
  generalizes a classical result about apparent singularities of linear
  differential equations, and it gives rise to a surprisingly simple
  desingularization algorithm.
\end{abstract}

 \maketitle

\section{Introduction}

Consider a linear ordinary differential equation, like for example
\[
  x(1-x) f'(x) - f(x) = 0.
\]
The leading coefficient polynomial $x(1-x)$ of the equation is of
special interest because every point $\xi$ which is a singularity of some
solution of the differential equation is also a root of this polynomial.
However, the converse is in general not true. In the example above, the
root $\xi=1$ indicates the singularity of the solution $x/(1-x)$, but there
is no solution which has a singularity at the other root $\xi=0$. To see this,
observe that after differentiating the equation, we can cancel (``remove'')
the factor~$x$ from it. The result is the higher order equation
\[
   (1-x)f''(x)-2f'(x)=0,
\]
whose solution space contains the solution space of the original equation. Such
a calculation is called \emph{desingularization.} The factor $x$ is said to be
\emph{removable.}

Given a differential equation, it is of interest to decide which factors of its
leading coefficient polynomial are removable, and to construct a higher order
equation in which all the removable factors are removed. A classical algorithm,
which is known since the end of the 19th century~\cite{schlesinger95,ince26}, proceeds by taking
the least common left multiple of the given differential operator with a
suitably constructed auxiliary operator. This algorithm is summarized in
Section~\ref{ssec:lclm-differential} below. At the end of the 20th century, the
corresponding problem for linear recurrence equations was studied and algorithms
for identifying removable factors have been found and their relations to
``singularities'' of solutions have been investigated~\cite{abramov99b,abramov03,abramov06b}. Also some
steps towards a unified theory for desingularization of Ore operators have been
made~\cite{chyzak10,chen13}. 
Possible connections to Ore closures of an operator ideal have been noted in~\cite{chyzak10} 
and within the context of order-degree curves~\cite{chen13,chen12c,chen12b}.
These will be further developed in a future paper.

Our contribution in the present article is a three-fold generalization of the
classical desingularization algorithm for differential equations. Our main
result (Theorem~\ref{thm:main} below) says that (a)~instead of the particular
auxiliary operator traditionally used, almost every other operator of
appropriate order also does the job, (b)~also the case when a multiple root of
the leading coefficient can't be removed completely but only its multiplicity
can be reduced is covered, and (c)~the technique works not only for differential
operators but for every Ore algebra.

For every removable factor~$p$ there is a smallest $n\in\set N$ such that
removing $p$ from the operator requires increasing the order of the operator by
at least~$n$. Classical desingularization algorithms compute for each factor~$p$
an upper bound for this~$n$, and then determine whether or not it is possible to
remove $p$ at the cost of increasing the order of the operator by at most~$n$.
In the present paper, we do not address the question of finding bounds on~$n$
but only discuss the second part: assuming some $n\in\set N$ is given as part of
the input, we consider the task of removing as many factors as possible without
increasing the order of the operator by more than~$n$. Of course, for Ore
algebras where it is known how to obtain bounds on~$n$, these bounds can be
combined with our result.

Recall the notion of Ore algebras~\cite{Ore33}. Let $K$ be a
field of characteristic zero. Let $\sigma\colon K[x]\to K[x]$ be a ring
automorphism that leaves the elements of $K$ fixed, and let $\delta\colon
K[x]\to K[x]$ be a $K$-linear map satisfying the law
$\delta(uv)=\delta(u)v+\sigma(u)\delta(v)$ for all $u,v\in K[x]$.  The algebra
$K[x][\partial]$ consists of all polynomials in $\partial$ with coefficients
in~$K[x]$ together with the usual addition and the unique (in general
noncommutative) multiplication satisfying $\partial
u=\sigma(u)\partial+\delta(u)$ for all $u\in K[x]$ is called an \emph{Ore
  algebra.} The field $K$ is called the \emph{constant field} of the algebra.
Every nonzero element $L$ of an Ore algebra $K[x][\partial]$ can be written uniquely
in the form
\[
 L = \ell_0 + \ell_1\partial + \cdots + \ell_r\partial^r
\]
with $\ell_0,\dots,\ell_r\in K[x]$ and $\ell_r\neq0$. We call
$\deg_\partial(L):=r$ the \emph{order} of~$L$ and $\lc_\partial(L):=\ell_r$ the
\emph{leading coefficient} of~$L$. Roots of the leading coefficient $\ell_r$ are
called singularities of~$L$.  Prominent examples of Ore algebras are the algebra
of linear differential operators (with $\sigma=\mathrm{id}$ and $\delta=\frac d{dx}$;
we will write $D$ instead of $\partial$ in this case) and the algebra of
linear recurrence operators (with $\sigma(x)=x+1$ and $\delta=0$; we will write $S$
instead of $\partial$ in this case).

We shall suppose that the reader is familiar with these definitions and facts,
and will make free use of well-known facts about Ore algebras, as explained, for
instance, in~\cite{Ore33, bronstein96,abramov05}. In particular, we will make use of the notion of
least common left multiples (lclm) of elements of Ore algebras: $L\in
K(x)[\partial]$ is a \emph{common left multiple} of $P,Q\in K(x)[\partial]$
if we have $L=UP=VQ$ for some $U,V\in K(x)[\partial]$, it is called a
\emph{least common left multiple} if there is no common left multiple of lower
order. Least common left multiples are unique up to left-multiplication by
nonzero elements of~$K(x)$. By $\lclm(P,Q)$ we denote a least common left
multiple whose coefficients belong to $K[x]$ and share no common divisors
in~$K[x]$. Note that $\lclm(P,Q)$ is unique up to (left-)multiplication by
nonzero elements of~$K$. Efficient algorithms for computing least common left
multiples are available~\cite{bostan12b}.

\section{The Differential Case} \label{ssec:lclm-differential}

In order to motivate our result, we begin by recalling the classical results
concerning the desingularization of linear differential operators. See the appendix
of~\cite{abramov06b} for further details on this case.

Let $L=\ell_0+\ell_1D+\cdots+\ell_rD^r\in K[x][D]$ be a differential operator of
order~$r$. Consider the power
series solutions of~$L$.  It can be shown that $x\nmid\ell_r$ if and only if $L$
admits $r$ power series solutions of the form $x^\alpha+\cdots$, for
$\alpha=0,\dots,r-1$. Therefore, if $x\mid\ell_r$, then this factor is removable if
and only if there exists some left multiple $M$ of~$L$, say with
$\deg_\partial(M)=s$, such that $M$ admits a power series solution with minimal exponent~$\alpha$
for every $\alpha=0,\dots,s-1$. This is the case if and only if $L$ has
$r$~linearly independent power series solutions with integer exponents
$0\leq\alpha_1<\alpha_2<\dots<\alpha_r$, because in this case (and only in this case) we can
construct a left multiple $M$ of $L$ with power series solutions~$x^\alpha+\cdots$
for each $\alpha=0,\dots,\max\{\alpha_1,\dots,\alpha_r\}-1$, by adding power series of
the missing orders to the solution space of~$L$.

These observations suggest the following desingularization algorithm for
operators $L\in K[x][\partial]$ with $x\mid\lc_\partial(L)$. First find the set
$\{\alpha_1,\dots,\alpha_\ell\}\subseteq\set N$ of all exponents~$\alpha_i$ for which there
exist power series solutions $x^{\alpha_i}+\cdots$. If $\ell<r$, return ``not desingularizable'' and
stop. Otherwise, let $m=\max\{\alpha_1,\dots,\alpha_\ell\}$ and let
$e_1,e_2,\dots,e_{m-\ell}$ be those nonnegative integers which are at most~$m$ but
not among the~$\alpha_i$. Return the operator
\[
  M = \lclm(L, A),
\]
where
\[
  A := \lclm(xD-e_1,\ xD-e_2,\ \dots,\ xD-e_{m-\ell}).
\]
Note that among the solutions of $A$ there are the monomials $x^{e_1},x^{e_2},\dots,x^{e_{m-\ell}}$,
and that the solutions of $M$ are linear combinations of solutions of $A$ and solutions of~$L$.
Therefore, by the choice of the $e_j$ and the remarks made above, $M$~is desingularized.

\begin{example}\label{ex:diff} Consider the operator
  \[
    L = (x-1) (x^2-3 x+3) x D^2
       -(x^2-3) (x^2-2 x+2) D
       +(x-2) (2 x^2-3 x+3)\in K[x][D].
  \]
  This operator has power series solutions with minimal exponents $\alpha=0$ and $\alpha=3$.
  Their first terms are
  \begin{alignat*}1
    &1 + x + \tfrac12x^2 -\tfrac18x^4 -\tfrac{19}{120}x^5 -\tfrac{119}{720}x^6 + \cdots,\\
    &x^3 + x^4 + x^5 + x^6 + \cdots.
  \end{alignat*}
  The missing exponents are $e_1=1$ and $e_2=2$. Therefore we take
  \[
    A:=\lclm(xD-1,xD-2)=x^2 D^2 - 2x D + 2
  \]
  and calculate
  \begin{alignat*}1
    M = \lclm(L,A) &=
  (x^5-2 x^4+4 x^3-9 x^2+12 x-6) D^4\\
&-(x^5-2 x^4+x^3-12 x^2+24 x-24) D^3\\
&-(3 x^3+9 x^2) D^2
+(6 x^2+18 x) D
-(6 x+18).
  \end{alignat*}
  Note that we have $x\nmid\lc_\partial(M)$, as predicted.
\end{example}

In the form sketched above, the algorithm applies only to the
singularity~$0$. In order to get rid of a different singularity, move this
singularity to $0$ by a suitable change of variables, then proceed as described
above, and after that undo the change of variables. Note that by removing the
singularity~$0$ we will in general introduce new singularities at other points.

\section{Removable Factors}

We now turn from the algebra of linear differential operators to arbitrary Ore algebras.
In the general case, removability of a factor of the leading coefficient is defined as follows.

\begin{definition}\label{def:removable}
Let $L\in K[x][\partial]$ and let $p\in K[x]$ be such that
$p\mid\lc_\partial(L)\in K[x]$. We say that $p$ is \emph{removable} from $L$ at
order~$n$ if there exists some $P\in K(x)[\partial]$ with $\deg_\partial(P)=n$
and some $v,w\in K[x]$ with $\gcd(p,w)=1$ such that $PL\in K[x][\partial]$ and
$\sigma^{-n}(\lc_\partial(PL))=\frac{w}{vp}\lc_\partial(L)$.  We then call $P$ a
\emph{$p$-removing} operator for~$L$, and $PL$ the corresponding
\emph{$p$-removed} operator.  $p$~is simply called \emph{removable} from $L$ if
it is removable at order~$n$ for some $n\in\set N$.
\end{definition}

\begin{example}
\begin{enumerate}
\item In the example from the introduction, we have $L=x(1-x)D-1\in K[x][D]$.
  An $x$-removing operator is $P=\frac1xD$: we have $PL=(1-x)D^2-2D$.
  Because of $\deg_\partial(P)=1$ we say that $x$ is removable at order~1.

  If $P$ is a $p$-removing operator then so is~$QP$, for every $Q\in K[x][\partial]$
  with $\gcd(\lc_\partial(Q),\sigma^{\deg_\partial(P)+\deg_\partial(Q)}(p))=1$.
  In particular, note that the definition permits to introduce some new factors~$w$
  into the leading coefficient while $p$ is being removed.
  For instance, in our example also $\frac{2-3x}{x}D$ is an $x$-removing operator
  for~$L$.
\item The definition does not imply that the leading coefficient of a
  $p$-removed operator is coprime with (a shifted copy of)~$p$. In general, it
  only requires that the multiplicity is reduced. As an example, consider the
  operator
  \[
    L= x^2(x-2)(x-1)D^2+2x(x^2-3 x+1) D-2 \in K[x][D]
  \]
  and $p=x$. The operator $P= \frac{x^4-x^3-4 x^2+2 x-2}{(x-2) x}D-(x^2+5 x+3) \in
  K(x)[D]$ is a $p$-removing operator because the leading coefficient of
  \begin{alignat*}1
    PL&= x(x-1) (x^4-x^3-4 x^2+2 x-2) D^3\\
    &\quad{}-(x^6-4 x^5-x^4+22 x^3-18 x^2+18 x-6) D^2\\
    &\quad{}-2 (x^5-x^4-8 x^3+8 x^2-3 x+6) D\\
    &\quad{}+2 (x^2+5 x+3)
  \end{alignat*}
  contains only one copy of $p$ while there are two of them in~$L$.
  This is called \emph{partial desingularization.}
  Observe that the definition permits to remove some factors~$v$ from the leading
  coefficient in addition to~$p$.
\item In the shift case, or more generally, in an Ore algebra where $\sigma$ is
  not the identity, the leading coefficient changes when an operator is
  multiplied by a power of~$\partial$ from the left. The application of $\sigma^{-n}$
  in the definition compensates this change. As an example, consider the operator
  \begin{alignat*}1
    L &= x (x+1) (5 x-2) S^2 -2 x (5 x^2-2 x-9) S\\
    &\quad{}+(x-4) (x+2) (5 x+3) \in K[x][S]
  \end{alignat*}
  and $p=x+1$. The operator $P=\frac{5 x^3+13 x^2-18 x-24}{(x+2) (5
    x+3)}S-\frac{2(5 x^3+28 x^2+23 x-24)}{(x+2) (5 x+3)}$ is a $p$-removing
  operator because the leading coefficient of
  \begin{alignat*}1
    PL&=(x+1)(5 x^3+13 x^2-18 x-24) S^3\\
    &\quad{} -2 (x+1)(10 x^3+21 x^2-58 x+24) S^2\\
    &\quad{}  +(25 x^4+60 x^3-217 x^2-84 x+288) S\\
    &\quad{} -2 (x-4)(5 x^3+28 x^2+23 x-24)
  \end{alignat*}
  does not contain $\sigma(p)=x+2$. It is irrelevant that it contains~$x+1$.
\end{enumerate}
\end{example}

As indicated in the examples, when removing a factor~$p$ from an
operator~$L$, Def.~\ref{def:removable} allows that we introduce other
factors~$w$, coprime to~$p$.  We are also always allowed to remove additional
factors~$v$ besides~$p$.  The freedom for having $v$ and $w$ is convenient but
not really necessary. In fact, whenever there exists an operator $P\in
K(x)[\partial]$ of order~$n$ such that $\sigma^{-n}(\lc_\partial(PL))=\frac
w{vp}\lc_\partial(L)$, then there also exists an operator $Q\in
K(x)[\partial]$ of order~$n$ such that $\sigma^{-n}(\lc_\partial(QL))=\frac
1p\lc_\partial(L)$.  To see this, note that by the extended Euclidean algorithm
there exist $s,t\in K[x]$ such that $sw+tp=1$. Set $Q=\sigma^n(sv)P +
\sigma^{-n}(t)\partial^n$. Then
\begin{alignat*}1
  \sigma^{-n}(\lc_\partial(QL))
&=sv\,\sigma^{-n}(\lc_\partial(PL))+ t\lc_\partial(\partial^n L)\\
&=sv\frac{w}{vp}\lc_\partial(L) + \frac{tp}{p}\lc_\partial(L)
=\frac1p\lc_\partial(L),
\end{alignat*}
as desired. This argument is borrowed from~\cite{abramov06b}.
The same argument can also be used to show the existence of operators that remove
all the removable factors in one stroke:

\begin{lemma}\label{lemma:everything}
Let $L\in K[x][\partial]$, let $n\in\set N$, and let
$\lc_\partial(L)=p_1^{e_1}p_2^{e_2}\cdots p_m^{e_m}$ be a factorization of the
leading coefficient into irreducible polynomials.
For each $i=1,\dots,m$, let $k_i\leq e_i$ be maximal such that $p_i$ is removable from $L$ at order~$n$.
Then there exists an operator $P\in K(x)[\partial]$ of order~$n$ such that
$\sigma^{-n}(\lc_\partial(PL))=\frac1{p_1^{k_1}p_2^{k_2}\cdots
  p_m^{k_m}}\lc_\partial(L)$.
\end{lemma}
\begin{proof}
  By the remark preceding the lemma, we may assume that for every $i$ there exists
  an operator $P_i\in K(x)[\partial]$ of order~$n$ with $P_iL\in K[x][\partial]$ and
  $\sigma^{-n}(\lc_\partial(P_iL))=p_i^{-k_i}\lc_\partial(L)$ (i.e., $w=v=1$).

  Next, observe that when $p$ and $q$ are two coprime factors of
  $\lc_\partial(L)$ which both are removable at order~$n$, then also their
  product $pq$ is removable at order~$n$.  Indeed, if $P,Q\in K(x)[\partial]$
  are such that $\deg_\partial(P)=\deg_\partial(Q)=n$, $PL,QL\in
  K[x][\partial]$, $\sigma^{-n}(\lc_\partial(PL))=\frac1p\lc_\partial(L)$, and
  $\sigma^{-n}(\lc_\partial(QL))=\frac1q\lc_\partial(L)$, and if $s,t\in K[x]$
  are such that $sq+tp=1$, then for $R:=\sigma^{-n}(s)P+\sigma^{-n}(t)Q$ we have
  $\sigma^{-n}(\lc_\partial(RL))=\frac1{pq}\lc_\partial(L)$, as desired.

  The claim of the lemma now follows by induction on~$m$, taking
  $p=p_1^{e_1}\cdots p_{m-1}^{e_m}$ and $q=p_m^{e_m}$.
\end{proof}

\section{Desingularization by Taking Least Common Left Multiples}

As outlined in Section~\ref{ssec:lclm-differential}, the classical algorithm for
desingularizing differential operators relies on taking the lclm of the operator
to be desingularized with a suitably chosen auxiliary operator. Our contribution
consists in a three-fold generalization of this approach: first, we show that it
works in every Ore algebra and not just for differential operators, second, we
show that almost every operator qualifies as an auxiliary operator in the lclm and
not just the particular operator used traditionally, and third, we show that the
approach also covers partial desingularization. From the second fact it follows
directly that taking the lclm with a random operator of appropriate order
removes, with high probability, \emph{all} the removable singularities of the
operator under consideration and not just a given one.

Consider an operator $L\in K[x][\partial]$ in an arbitrary Ore algebra, and let
$p\mid\lc_\partial(L)$ be a factor of its leading coefficient. Assume that this
factor is removable at order~$n$. Our goal is to show that for almost all
operators $A\in K[\partial]$ of order~$n$ with constant coefficients the
operator $\lclm(L,A)$ is $p$-removed.

One way of computing the least common left multiple of two operators $L,A\in
K[x][\partial]$ with $\deg_\partial(L)=r$ and $\deg_\partial(A)=n$ is as
follows. Make an ansatz with undetermined coefficients $u_0,\dots,u_n$,
$v_0,\dots,v_r$ and compare coefficients of $\partial^i$ ($i=0,\dots,n+r$) in
the equation
\[
  (u_0+\cdots+u_{n-1}\partial^{n-1}+u_n\partial^n)L =
  (v_0+\cdots+v_{r-1}\partial^{r-1}+v_r\partial^r)A.
\]
This leads to a system of homogeneous linear equations over $K(x)$ for the
undetermined coefficients, which has more variables than equations and therefore
must have a nontrivial solution. For each solution, the operator on either side
of the equation is a common left multiple of $L$ and~$A$.

For most choices of $A$ the solution space will have dimension~$1$, and in this
case, for every nontrivial solution we have $u_n\neq0$. In particular the least
common left multiple $M=\lclm(L,A)$ has then order~$r+n$. The singularities of
$M$ are then the roots of $\sigma^n(\lc_\partial(L))$ plus the roots $u_n$ minus
the common roots of $u_0,\dots,u_n$, which are cancelled out by convention. It
is not obvious at this point why removable factors should appear among the
common factors of $u_0,\dots,u_n$. To see that they systematically do, consider
a $p$-removing operator $P\in K(x)[\partial]$ of order~$n$, and observe that the
operators $1,\partial,\dots,\partial^{n-1},\partial^n$ generate the same
$K(x)$-vector space as $1,\partial,\dots,\partial^{n-1},P$. If we use the latter
basis in the ansatz for the lclm, i.e., do coefficient comparison in
\[
  (u_0+\cdots+u_{n-1}\partial^{n-1} + u_nP)L =
  (v_0+\cdots+v_{r-1}\partial^{r-1}+v_r\partial^r)A,
\]
then every nontrivial solution vector $(u_0,\dots,u_n,v_0,\dots,v_r)$ of the
resulting linear system gives rise to a common left multiple of $L$ and $A$
in $K[x][\partial]$ whose singularities are the roots of
$\lc_\partial(PL)=\sigma^n(\frac1p\lc_\partial(L))$ plus the roots of $u_n$
minus the common roots of~$u_0,\dots,u_n$.
This argument shows that the removable factor $p$ will have disappeared in the
lclm \emph{unless} it is reintroduced by~$u_n$. The main technical difficulty
to be addressed in the following is to show that this can happen only for very
special choices of~$A$. For the proof of this result we need the following
lemma.

\begin{lemma}\label{lemma:det}
  Let $n,m\in\set N$,
  let $v_1,\dots,v_n\in K^{n+m}$ be linearly independent over~$K$, and
  let $w_1,\dots,w_m\in K[x_1,\dots,x_n]^{n+m}$ be defined by
  \begin{alignat*}1
    w_1&=(x_1,\dots,x_n,1,0,\dots,0)\\
    w_2&=(0,x_1,\dots,x_n,1,0,\dots,0)\\
       &{}\vdots\\
    w_m&=(0,\dots,0,x_1,\dots,x_n,1).
  \end{alignat*}
  Then $\Delta:=\det(w_1,\dots,w_m,v_1,\dots,v_n)$ is a nonzero polynomial in $K[x_1,\dots,x_n]$.
\end{lemma}
\begin{proof}
Simultaneous induction on $n$ and~$m$: We show that the lemma holds for $(n,m)$
if it holds for $(n-1,m)$ and for $(n,m-1)$.

As induction basis, observe that the lemma holds for $n=1$, $m$ arbitrary, and
also for $n$ arbitrary,~$m=1$.

Now let $(n,m)\in\set N^2$ with $n\geq2,m\geq2$ be given. Let $v_1,\dots,v_n\in K^{n+m}$ be linearly
independent. Write $v_i=(v_{1,i},\dots,v_{n+m,i})$ for the coefficients.

Case~1. $v_{1,1}=v_{1,2}=\cdots=v_{1,n}=0$. In this case, the vectors $\bar v_i\in K^{n+(m-1)}$
obtained from the $v_i$ by chopping the first coordinate must be linearly independent.
By expanding along the first row, we have
\[
  \Delta=x_1\det(\bar w_2,\dots,\bar w_m,\bar v_1,\dots,\bar v_n).
\]
The determinant on the right is nonzero by applying the lemma with $n$ and $m-1$. Therefore
the determinant on the left is also nonzero.

Case~2. If at least one of the $v_{1,j}$ is nonzero, then we may assume without loss of generality
that $v_{1,1}=1$ and $v_{1,2}=v_{1,3}=\cdots=v_{1,n}=0$, by performing suitable column operations on
$(v_1,\dots,v_n)\in K^{(n+m)\times n}$. Then the vectors
$\bar v_2,\dots,\bar v_n\in K^{(n-1)+m}$ obtained from the $v_i$ by chopping the first coordinate
are linearly independent. Expanding along the first row, we now have
\[
  \Delta= x_1\, [[\textrm{poly}]] + v_{1,1}
  \begin{vmatrix}
    x_2 & x_1 & 0 & \cdots & 0 & v_{2,2} & \cdots & v_{2,n} \\
    x_3 & x_2 & \ddots & \ddots & \vdots & \vdots & &\vdots \\
    \vdots & \ddots & \ddots & \ddots & 0 & \vdots & &\vdots\\
    \vdots & & \ddots & \ddots & x_1 & \vdots & &\vdots\\
    x_n &  & & \ddots & x_2 & \vdots & &\vdots\\
    1 & \ddots &  & & x_3 & \vdots & &\vdots\\
    0 & \ddots & \ddots &  & \vdots & \vdots & &\vdots\\
    \vdots & \ddots & \ddots & \ddots  & \vdots & \vdots & &\vdots\\
    \vdots &        & \ddots & \ddots   & x_n & \vdots & &\vdots\\
    0 & \cdots &  \cdots & 0 & 1 & v_{n+m,2} & \cdots & v_{n+m,n} \\
  \end{vmatrix}.
\]
By setting $x_1=0$, the first term on the right hand side disappears, and so do
the entries $x_1$ in the determinant of the second term. By induction
hypothesis, the determinant on the right with $x_1$ set to zero is a nonzero
polynomial in $x_2,\dots,x_n$.  Since also $v_{1,1}\neq0$, the whole right hand
side is nonzero for $x_1=0$.  Consequently, when $x_1$ is not set to zero, it
cannot be the zero polynomial.
\end{proof}

\begin{theorem}[Main result]\label{thm:main}
  Let $K[x][\partial]$ be an Ore algebra, let
  $L\in K[x][\partial]$ be an operator of order~$r$, and let $n\in\set N$.
  Let $p\in K[x][\partial]$ be an irreducible polynomial
  which appears with multiplicity $e$ in~$\lc_\partial(L)$
  and let $k\leq e$ be maximal such that $p^k$ is removable from $L$ at order~$n$.
  Let $A=a_0+a_1\partial+\cdots+a_{n-1}\partial^{n-1}+\partial^n$ in
  $K[a_0,\dots,a_{n-1}][\partial]$, where $a_0,\dots,a_{n-1}$ are new
  constants, algebraically independent over~$K$.
  Then the multiplicity of $\sigma^n(p)$ in $\lc_\partial(\lclm(L,A))$ is $e-k$.
\end{theorem}
\begin{proof}
  Let $P_0,\dots,P_n\in K(x)[\partial]$ be such that each $P_i$ has order~$i$
  and removes from $L$ all the factors of $\lc_\partial(L)$ that can possibly be removed
  by an operator of order~$i$. Such operators exist by
  Lemma~\ref{lemma:everything}. Consider an ansatz
\[
  u_0P_0L+u_1P_1L+\cdots+u_nP_nL
= v_0A + v_1\partial A + \cdots + v_r\partial^r A
\]
with unknown $u_i,v_j\in K[a_0,\dots,a_{n-1}][x]$. Compare coefficients with respect to
powers of $\partial$ on both sides and solve the resulting linear system. This gives
a polynomial solution vector with
\[
  u_n = \det\bigl(
    [P_0L], [P_1L], \cdots [P_{n-1}L], [A], [\partial A],  \cdots,  [\partial^{r-1}A]\bigr),
\]
where the notation $[U]$ refers to the coefficient vector of the operator~$U$ (padded with
zeros, if necessary, to dimension $r+n$).

If $\sigma^n(p)\mid u_n$, then the columns of the determinant are linearly dependent when
viewed as elements of $F[a_0,\dots,a_{n-1}]$ with $F=K[x]/\<\sigma^n(p)>$. Then Lemma~\ref{lemma:det}
with $F$ in place of $K$ implies that already $[P_0L],\dots,[P_{n-1}L]$ are
linearly dependent modulo~$\sigma^n(p)$. In other words, there are polynomials
$u_0,\dots,u_{n-1}\in K[x]$ of degree $<\deg(p)$, not all zero, such that the
linear combination $u_0P_0L + \cdots + u_{n-1}P_{n-1}L$ has content~$\sigma^n(p)$.  If $d$ is
maximal such that $u_d\neq0$, then this means that $\frac1{\sigma^d(p)}(u_0P_0+\cdots+u_dP_d)$ is an
operator of order~$d$ which removes from $L$ one factor $\sigma^{n-d}(p)$ more
than $P_d$ does, in contradiction to the assumption that $P_d$ removes as much
as possible.
\end{proof}

The theorem continues to hold when the indeterminates $a_0,\dots,a_{n-1}$ are
replaced by values in $K$ which do not form a point on the zero set of the
determinant polynomial $u_n\bmod\sigma^n(p)$, as discussed in the proof. As this
is not the zero polynomial and we assume throughout that $K$ has characteristic
zero, it follows that almost all choices of $A\in K[\partial]$
will successfully remove all the factors of $\lc_\partial(L)$ that are removable
at order~$\deg_\partial(A)$.

The theorem thus gives rise to the following very simple probabilistic algorithm
for removing, with high probability, as many factors as possible from a given
operator $L\in K[x][\partial]$ at a given order~$n$:
\begin{itemize}
\item Pick an operator $A\in K[\partial]$ of order~$n$ at random.
\item Return $\lclm(L,A)$.
\end{itemize}
This is a Monte Carlo algorithm: it always terminates but with low probability
may return an incorrect answer. For a Las Vegas algorithm (low probability of
not terminating but every answer is guaranteed to be correct), inspect the
multiplier $u_n$ which appears during the construction of the lclm: if it is
coprime with $\sigma^n(\lc_\partial(L))$, then no removed singularities get
mistakenly re-introduced and the result is therefore correct. Otherwise, try
again. For a deterministic algorithm, don't take the operators $A$ at random but
from some enumeration of $K[\partial]$ which is chosen in such a way that the
Zariski closure of the set of the corresponding coefficient vectors is all
of~$K^n$.

The Monte Carlo version of the algorithm is included in the new
\verb|ore_algebra| package for Sage~\cite{kauers14}, and works very efficiently
thanks to the efficient implementation of least common left multiples also
available in this package.  This package has been used for the calculations in
the following concluding examples. The computation time for all these examples
is negligible.

\begin{example}
\begin{enumerate}
\item For $L\in\set Q[x][D]$ from Example~\ref{ex:diff} and the ``randomly chosen'' operator
  $A=D^2+D+1$ we have
  \begin{alignat*}1
    \lclm(L,A)&=(x^7-4 x^6+6 x^5-4 x^4+x^3+6 x-6) D^4 \\
              &-(2 x^6-9 x^5+15 x^4-11 x^3+3 x^2-24)D^3\\
              &-(x^7-4 x^6+6 x^5-4 x^4+x^3+6 x-6) D\\
              &+(2 x^6-9 x^5+15 x^4-11 x^3+3 x^2-24).
  \end{alignat*}
  This is not the same result as in Example~\ref{ex:diff}, but it does have the required
  property $x\nmid\lc_\partial(\lclm(L,A))$.
\item This is an example for the recurrence case. Let
  \begin{alignat*}1
\qquad L &= 2 (x+3)^2 (59 x+94) S^3
-(2301 x^3+15171 x^2+32696 x+22876) S^2\\
&\quad{}-5 (59 x^3+330 x^2+600 x+359) S
-(59 x+153) (x+1)^2.
  \end{alignat*}
  Among the factors of $(x+3)$ and $(59x+94)$ of the leading coefficient, the latter is
  removable at order~1 and the former is not removable.
  Accordingly, for the ``randomly chosen'' operator $A=S-2$ we have
  \begin{alignat*}1
   \qquad\lclm(L,A) &= 2  (x+4)^2 (8909 x^3+57087 x^2+119629 x+81711) S^4 \\
     &+ ({\cdots}) S^3 + ({\cdots}) S^2 + ({\cdots}) S + ({\cdots}),
  \end{alignat*}
  where $({\cdots})$ stands for some other polynomials.
  Note that the leading coefficient is coprime to $\sigma(59x+94)=59 x+153$.
\item As an example for partial desingularization, consider the operator
  $L=x^3D^3 - 3x^2D^2-2xD+10\in\set Q[x][D]$. Of the three copies of~$x$ in the leading coefficient,
  one is removable at order~2, another one at order~4, and the third is not removable.
  In perfect accordance, we find for example
  \begin{alignat*}1
    \qquad
    &\lc_\partial(\lclm(L, D+2))=x^3(4x^3+6x^2-2x-5),\\
    &\lc_\partial(\lclm(L, D^2+1))=x^2(x^6+10x^4+40x^2+80),\\
    &\lc_\partial(\lclm(L, D^3+3D^2-1))=x^2(x^8-30x^6+\cdots+2160x+1920),\\
    &\lc_\partial(\lclm(L, D^4-D^2+1))=x(x^{10}-10x^8+120x^6-720x^4-3200),\\
    &\lc_\partial(\lclm(L, D^5+D-1))=x(x^{12}-3x^{11}+\cdots+25600x-22400).
  \end{alignat*}
\item There are unlucky choices for~$A$. For example, consider
  \begin{alignat*}1
   L &= (x-7)(x^2-2x-12)S^2 - (3x^3-23x^2-23x+291)S\\
   &\quad{}+ 2(x-6)(x^2-13)\in\set Q[x][S].
  \end{alignat*}
  The factor $x-7$ is removable, as can be seen, for example, from
  the fact that $\lc_\partial(\lclm(L,S-1))=2 x^2-x-51$ is coprime to $\sigma(x-7)=x-6$.
  However, if we take $A=S-\tfrac94$, then
 \begin{alignat*}1
   &\lclm(L,A)=
  4 (x-7) (x-6) (5 x-28) S^3\\
 &\qquad{} - (x-7) (3092 - 1138 x + 105 x^2) S^2\\
 &\qquad{} + (x-5) (6081 - 2080 x + 175 x^2) S\\
 &\qquad{} - 18 (x-6) (x-5) (5 x-23),
 \end{alignat*}
 which has $x-6$ in the leading coefficient. (It is irrelevant that also $x-7$ appears as
 a factor.)
\item Finally, as an example with an unusual Ore algebra,
  consider $\set Q[x][\partial]$ with $\sigma\colon\set Q[x]\to\set Q[x]$
  defined by $\sigma(x)=x^2$ and $\delta\colon\set Q[x]\to\set Q[x]$ defined by $\delta(x)=1-x$.
  Let
  \begin{alignat*}1
    L &= (2x+1)\partial^2 + (x^2+3x-1)\partial - (2x^4+2x^3+x^2+1).
  \end{alignat*}
  The factor $2x+1$ is removable at order~$1$. For example, for $A=\partial-1$ we find
  that $\lclm(L,A)$ equals
  \begin{alignat*}1
\qquad&(2x^3 + 4x^2 + 4x - 1)\partial^3 - (2x^6 - x^4 - 4x^3 - 3x^2 + x + 5)\partial^2 \\
&\quad{}- (2x^9 + 4x^8 + 6x^7 + 4x^6 + 2x^5 + 3x^4 + 2x^3 + 3x^2 + 3x - 2)\partial\\
&\quad{}+ (2x^9 + 4x^8 + 6x^7 + 6x^6 + 2x^5 + 2x^4 - 4x^3 - 4x^2 + 4).
  \end{alignat*}
  As expected, the leading coefficient does not contain the factor $\sigma(\lc_\partial(L))=2x^2+1$.
\end{enumerate}
\end{example}

\bibliographystyle{plain}
\bibliography{bib}

\end{document}